\documentclass[11pt,a4paper]{article}

\usepackage{amssymb}
\usepackage{amsmath}
\usepackage{amsthm}
\usepackage{textcomp}
\usepackage{array}
\usepackage{lineno}
\usepackage{color}
\usepackage[margin=0.92in]{geometry}
\usepackage{cite}

\usepackage{graphicx}

\newcommand{\bra}[1]{\langle #1|}
\newcommand{\ket}[1]{|#1\rangle}
\newcommand{\braket}[2]{\langle #1|#2\rangle}
\newcommand{\cent}[0]{\mbox{\textcent}}
\newcommand{\dollar}[0]{\$}

\def\Sol{{\sf Sol}}
\def\Par{{\sf Par}}
\def\Ind{{\sf Ind}}
\def\even{{\sf Even}}
\def\less{{\sf Less}}
\def\EQ{{\sf EQ}}
\def\dec{{\sf dec}}
\def\span{{\sf span}}

\newtheorem{fact}{Fact}
\newtheorem{theorem}{Theorem}
\newtheorem{corollary}{Corollary}
\newtheorem{lemma}{Lemma}
\newtheorem{remark}{Remark}

\title{Quantum, Stochastic, and Pseudo Stochastic Languages \\ with Few States%
  \thanks{A~preliminary version of this work is \cite{ShY14}}}%
\author{Arseny Shur\thanks{Partially supported under the Agreement 02.A03.21.0006 of 27.08.2013 between the Ministry of Education and Science of the Russian Federation and Ural Federal University.} 
    \\
    Ural Federal University, Ekaterinburg, Russia
    \\
    arseny.shur@urfu.ru
   \\ \\
 Abuzer Yakary{\i}lmaz\thanks{Partially supported by CAPES with grant 88881.030338/2013-01,  ERC Advanced Grant MQC, and FP7 FET projects QALGO.}
    \\
    National Laboratory for Scientific Computing,  Petr\'{o}polis, RJ, Brazil
    \\
    abuzer@lncc.br
}
    
\date{\today}%

\begin{document}

\maketitle

\begin{abstract}

Stochastic languages are the languages recognized by probabilistic finite automata (PFAs) with cutpoint over the field of real numbers. More general computational models over the same field such as generalized finite automata (GFAs) and quantum finite automata (QFAs) define the same class. In 1963, Rabin proved the set of stochastic languages to be uncountable presenting a single 2-state PFA over the binary alphabet recognizing uncountably many languages depending on the cutpoint. In this paper, we show the same result for unary stochastic languages. Namely, we exhibit a 2-state unary GFA, a 2-state unary QFA, and a family of 3-state unary PFAs recognizing uncountably many languages; all these numbers of states are optimal. After this, we completely characterize the class of languages recognized by 1-state GFAs, which is the only nontrivial class of languages recognized by 1-state automata. Finally, we consider the variations of PFAs, QFAs, and GFAs based on the notion of inclusive/exclusive cutpoint, and present some results on their expressive power.
\\
\textbf{keywords:} stochastic languages, unary languages, quantum finite automata, generalized finite automata, probabilistic finite automata, regular languages, context-free languages
\end{abstract}

\section{Introduction}

Computation models based on real, or even complex, numbers are much more powerful then ``classical'' Turing machines. Since there is a possibility that some of these models, like the quantum model, will become physically available for experiments in the nearest future, it is quite important to know the limitations of the models. In the paper, we study the power of small probabilistic, general, and quantum automata. The two main questions are how many states is sufficient to recognize uncountably many unary languages? what languages can be recognized by one state? Similar questions were studied starting from the seminal paper by Rabin \cite{Rab63}, but not all of them are answered yet.

Our results are as follows. In Sect.~\ref{sec:card-unary}, we first show that a rotation operator implemented by a 2-state unary GFA or QFA generates uncountably many languages depending on the choice of the cutpoint. For QFAs, the result holds even for the most restricted model of such an automata, described in \cite{MC00}. This fact also allows us to answer an open question stated in \cite{YS10A}. Since 1-state unary GFAs recognize only regular languages (see Sect.~\ref{sec:one-state-pseudo} for details), the obtained bounds on the number of states are sharp. Then we turn to PFAs, where the situation differs because (i) 2-state unary PFAs recognize only regular languages  and (ii) the choice of a cutpoint for a unary PFA gives only countably many distinct languages; see \cite{Paz71}. We exhibit an uncountable set of pairs (3-state unary PFA; cutpoint) producing uncountably many different languages. Again, the bound on the number of states is sharp.

1-state PFAs and QFAs define trivial languages but the situation is completely different for GFAs. In the unary case, 1-state GFAs recognize a proper subclass of regular languages, while the set of binary languages recognized by 1-state GFAs is uncountable. In Sect.~\ref{sec:one-state-pseudo}, we introduce three classes of languages (solution, parity, and indicator languages), fully characterize the languages recognized by 1-state GFAs in terms of these classes and provide criteria of regularity and context-freeness for these languages.

In the last part of the paper (Sect.~\ref{sec:incl-excl}), we consider GFAs/QFAs/PFAs using cutpoint in a different way. Namely, either equality or non-equality is used as the acceptance condition instead of the '$>$' inequality. We prove some results on the expressive power of automata with such acceptance conditions.

\section{Background}

We denote the set of states by $ Q = \{q_1,\ldots,q_n\} $ for some $n>0$ and the input alphabet by $ \Sigma $. The left end-marker $\cent$ and the right end-marker $\dollar$ do not belong to $\Sigma$. All models in the paper read inputs from the left to the right symbol by symbol. 

A generalized finite automaton (GFA) \cite{Tur69,Paz71} $\cal G$ is a quintuple
\[
	\mathcal{G} = (Q,\Sigma,\{ A_{\sigma} \mid \sigma \in \Sigma \},v_0,f),
\]
where $A_{\sigma} \in \mathbb{R}_{|Q|\times |Q|}$ is the \emph{transition matrix} for the symbol $\sigma \in \Sigma$, $v_0\in \mathbb{R}_{|Q|\times 1}$ is the \emph{initial vector}, and $f\in \mathbb{R}_{1\times |Q|}$ is the \emph{final vector}.
For a given input $w\in \Sigma^*$, the computation of $\cal G$ can be traced by a $ |Q| $-dimensional column vector:
\[
	v_{i} = A_{w_i} v_{i-1},
\] 
where $ 1 \leq i \leq |w| $ and the accepting value of $\cal G$ on $w$ is calculated as
\[
	f_{\cal G}(w) = f v_{|w|} = f A_{w_{|w|}} A_{w_{|w|-1}} \cdots A_{w_2} A_{w_1} v_0.
\]

A probabilistic finite automaton (PFA) \cite{Rab63} is a special case of GFA where each transition matrix is (left) stochastic, $ v_0 $ is a 0-1 stochastic vector, and $f$ is a 0-1 vector. Note that the entry of 1 in $v_0$ corresponds to a state called the initial state and the entries of 1s in $f$ correspond to the states called accepting (or final) states. 

A PFA can also be defined by starting its computation in a distribution of states instead of a single state. Then any stochastic vector can serve as the initial vector. Similarly, instead of some fixed accepting states, each state contributes to the accepting probability with some weight from $ [0,1] $. Formally, we can assume that a PFA can (i) read the left end-marker ($\cent$) before reading the input for preprocessing (and so the new initial vector is $ A_{\cent} v_0 $ for a stochastic matrix $A_{\cent}$) and (ii) read the right end-marker after finishing the whole input for post-processing (and so the new final vector is $ f A_{\dollar} $ for a stochastic matrix $A_{\dollar}$). 

In the literature, there are different models of quantum finite automata (QFAs). The most general one \cite{Hir10,YS11A} can simulate PFAs exactly (see \cite{SayY15} for a pedagogical proof).
In this paper, we mainly use the most restricted model called MCQFA\footnote{MC stands for Moore and Crutchfield who introduced the model \cite{MC00}.} \cite{MC00} which is sufficient to follow most of our quantum results. 

We begin with a concise review of quantum computation. Conventionally, in quantum computation (mechanics), any vector is represented in ``ket'' notation, e.g. $ \ket{v} $. Its conjugate transpose is denoted by $\bra{v}$ and the inner product of two vectors $\bra{u}$ and $\ket{v}$ is denoted by $\braket{u}{v}$. A quantum state of a quantum system $ \mathcal{M} $ with the set of states $Q=\{q_1,\ldots,q_n\}$ is a norm-1 (column) vector in the $n$-dimensional Hilbert space ${\cal H}_n$:
\[
	\ket{v} = \left( \begin{array}{c} \alpha_1 \\ \vdots \\ \alpha_n 
\end{array}	 \right),\ \text{ where }\ \sum_{j=1}^n  | \alpha_j | ^2 = 1.
\]
The entries $\alpha_1,\ldots,\alpha_n$ are called \emph{amplitudes} of the states $q_1,\ldots,q_n$, respectively, while $ |\alpha_j|^2 $ is viewed as the probability of the system being in the state $q_j$. The quantum state containing 1 in the $j$th entry (and hence zeroes in the other entries) is denoted by $\ket{q_j}$. Clearly, $\ket{q_1},\ldots,\ket{q_n}$ form a basis of ${\cal H}_n$. 

There are two fundamental quantum operations: \emph{unitary} and \emph{measurement} operators. A unitary operator applicable to $\mathcal{M}$ is an $n \times n$ complex-valued matrix preserving the norm. Let $\ket{v}$ be a quantum state satisfying $ \braket{v}{v} = 1 $ and $ U $ be a unitary operator. The new quantum state after applying $U$ is $\ket{v'} = U \ket{v}$.

Measurement operators are used to retrieve information from quantum systems. We use simple measurement operators defined as follows. The set of states is partitioned into sets $Q_1,\ldots,Q_k$ ($k>1$) inducing the decomposition of $\mathcal{H}_n$ into the sum $ \mathcal{H} = \mathcal{H}_1 \oplus \cdots \oplus \mathcal{H}_k $ of orthogonal subspaces $ \mathcal{H}_l = \span \{\ket{q} \mid q \in Q_l\} $. A measurement operator $P$ has $k$ operation elements $P_l = \sum_{q \in Q_l}  \ket{q}\bra{q}$ and forces the system to collapse into one of $k$ quantum subsystems corresponding to the subspaces ${\cal H}_l$. We denote the outcomes of $P$ with the indices ``$1$''$,\ldots,$``$k$''. The probability of getting the outcome ``$l$'' is 
\[
	p_l=\braket{\tilde{v_l}}{\tilde{v_l}} = \sum_{q_j \in Q_l} |\alpha_j|^2, \text{ where } \ket{\tilde{v_l}} = P_l v.
\]
If $\mathcal{M}$ collapses to this subsystem ($p_l>0$), the new quantum state is obtained by normalizing $\ket{\tilde{v_l}}$:
\[
	\ket{v_l} = \frac 1{\sqrt{p_l}} \ket{\tilde{v_l}}.
\]

A quantum system can also be in more than one quantum state, called \textit{pure state}, with some probabilities:
\[
	\left\lbrace
	 	\left(p_j,\ket{v_j}\right) ~ \middle| ~ p_j \in [0,1], \braket{v_j}{v_j} = 1, 1 \leq j \leq k, \sum_{j=1}^k p_j =1 
	 \right\rbrace
\]
A convenient way of representing such a mixture, called \textit{mixed state}, is using a density matrix (also called density operator):
\[
	\rho = \sum_{j=1}^k p_j \ket{v_j} \bra{v_j}.
\]
Any density matrix ($\rho$) satisfies three properties: (i) $Tr(\rho) = 1 $, (ii) it is Hermitian, and (iii) positive semi-definite. Note that the $j^{th}$ diagonal entry gives the probability of the system being in state $ \ket{q_j} $.  

The most general quantum operator which generalizes any stochastic and unitary operator is superoperator. Formally, a superoperator consists of $l>0$ operation elements
\[
	\mathcal{E} = \{ E_1,\ldots, E_l \} 
\]
satisfying 
\[
	\sum_{j=1}^l E_j^{\dagger} E_j = I.
\]
An easy way to determine whether a given operator ($\mathcal{E}$) is superoperator is as follows. Let $E$ be the following rectangular matrix
\[
	E = \left( \begin{array}{c}
		E_1 \\ \vdots \\ E_l
\end{array}	 \right).
\]
Then, the columns of $E$ form an orthonormal set if and only if $\mathcal{E}$ is superoperator. If the quantum system is in mixed state $ \rho $, then the new state, after applying superoperator $\mathcal{E}$, is 
\[
	\rho' = \mathcal{E}(\rho) = \sum_{j=1}^l E_j \rho E_j^{\dagger}.
\]
If the measurement operator $P = \{P_1,\ldots,P_k\} $ described above is applied to the state $\rho$, the outcome ``$j$'' is obtained with probability
\[	
	p_j = Tr(P_j \rho)
\]
and the new (normalized) state, if $ p_j >0 $, becomes
\[		
	\rho_j = \frac{P_j \rho}{\sqrt{p_j}}.
\]

A general measurement operator is a superoperator $ \mathcal{E} = \{E_1,\ldots,E_l\} $ where indices ``$1$''$,\ldots,$``$l$'' are measurement outcomes. For a given mixed (or pure) state $\rho$, the probability of obtaining outcome ``$j$'', say $p_j$, can be calculated as follows:
\[
	p_j = Tr ( \tilde{\rho_j} ), \mbox{ where } \tilde{\rho_j} = E_j \rho E_j^\dagger.
\]
If outcome ``$j$'' is observed ($p_j>0$), then the system collapses to 
\[
	\rho_j = \frac{ \tilde{\rho_j}}{p_j}.
\]

A MCQFA is a quintuple
$	\mathcal{M} = (Q,\Sigma,\{ U_{\sigma} \mid \sigma \in \Sigma \},\ket{v_0},P), $
where $Q=\{q_1,\ldots,q_n\}$, $U_{\sigma}\in \mathbb{C}_{|Q|\times |Q|} $ is the unitary transition matrix for the symbol $ \sigma \in \Sigma $, $\ket{v_0}\in \{\ket{q_1},\ldots,\ket{q_n}\}$ is the initial state, and $P = \{ P_a,P_r \}$ is the measurement operator applied after reading the whole input. An input is accepted if the outcome ``a'' of $P$ is observed.
For any given input $w\in \Sigma^*$, the computation of $\cal M$ can be traced by a $ |Q| $-dimensional quantum state:
\[
	\ket{v_{i}} = U_{w_i} \ket{v_{i-1}},
\] 
where $ 1 \leq i \leq |w| $. The accepting probability of $\cal M$ on $w$ is 
\[
	f_{\cal M}(w) =  \braket{\tilde{v_{a}}}{\tilde{v_{a}}},  \mbox{ where } \ket{\tilde{v_a}} = P_a \ket{v_{|w|}}.
\]
MCQFAs can also be defined with the end-markers to perform pre- and post-processing of the input. Then the initial state can be an arbitrary quantum state $ U_{\cent} \ket{v_0} $ for a unitary operator $U_{\cent}$, and the measurement turns out to be a  general one with two outcomes, $ \{P_a U_{\dollar},P_r U_{\dollar} \} $, for a unitary $U_{\dollar}$. On the other hand, any MCQFA with both end-markers can be equivalently represented by a MCQFA with a single end-marker \cite{BP02}. Therefore, any MCQFA with both end-markers can be defined like MCQFA without end-markers except that $ \ket{v_0} $ can be an arbitrary quantum state.

A (general) quantum finite automaton (QFA) \cite{Hir10,YS11A} is a quintuple
\[	
	\mathcal{M} = (Q,\Sigma,\{ \mathcal{E}_{\sigma} \mid \sigma \in \Sigma \},\ket{v_0},P), 
\]
where $Q=\{q_1,\ldots,q_n\}$, $\mathcal{E}_{\sigma} = \{ E_{\sigma,1},\ldots,E_{\sigma,l_\sigma} \} $ is the superoperator for the symbol $ \sigma \in \Sigma $ composed by $l_\sigma$ operation elements, $ \ket{v_0}\in \{\ket{q_1},\ldots,\ket{q_n}\}$ is the initial state, and $P = \{ P_a,P_r \}$ is the measurement operator applied after reading the whole input. An input is accepted if the outcome ``a'' of $P$ is observed.
For any given input $w\in \Sigma^*$, the computation of $\cal M$ can be traced by a $ |Q| \times |Q| $-dimensional density operator (mixed state):
\[
	\rho_j = \mathcal{E}_{w_i} (\rho_{j-1}),
\] 
where $\rho_0 = \ket{q_0} \bra{q_0} $ and $ 1 \leq j \leq |w| $, and the accepting probability of $\cal M$ on $w$ is 
\[
	f_{\cal M}(w) = Tr( P_a \rho_{|w|} ).
\]

QFAs can also be defined with the end-markers to perform pre- and post-processing of the input. Then the initial state can be an arbitrary mixed quantum state $ \mathcal{E}_{\cent} ( \rho_0 ) $ for a superoperator operator $\mathcal{E}_{\cent}$, and the measurement turns out to be a general one with two outcomes,

\[
	\{ \{ P_a E_{\dollar,1},\ldots,P_a E_{\dollar,l} \}, \{ P_r E_{\dollar,1},\ldots,P_r E_{\dollar,l} \} \},
\] 
for a superoperator $\mathcal{E}_{\dollar} = \{ E_{\dollar,1},\ldots, E_{\dollar,l} \}$.

The language recognized by GFA/PFA/QFA $\cal M$ with \emph{cutpoint} $\lambda $ is defined as
\[
	L(\mathcal{M},\lambda) = \{ w \in \Sigma^* \mid f_{\cal M}(w) > \lambda \},
\]
where $ \lambda\in \mathbb{R} $ for GFAs and in $\lambda\in [0,1) $ for PFAs and QFAs. Any such language recognized by an $n$-state GFA [PFA, QFA] is called ($n$-state) pseudo stochastic [resp., stochastic, quantum automaton] language. The class names are given below:
\[
	\begin{array}{|l|l|l|}
		\hline
		\multicolumn{1}{|c|}{\textbf{model}} & \mbox{\textbf{general alphabet}} & \mbox{\textbf{unary alphabet}}
		\\ \hline \hline
		\mbox{GFA} & \sf PseudoS & \sf UnaryPseudoS 
		\\ \hline
		\mbox{PFA} & \sf S & \sf UnaryS 
		\\ \hline
		\mbox{QFA} & \sf QAL & \sf UnaryQAL 
		\\ \hline
		\mbox{MCQFA} & \sf MCL & \sf UnaryMCL
		\\ \hline
	\end{array}
\]
For class $\mathsf{C}$, one can define a new class using up to three parameters in brackets $\mathsf{C[\cent n \dollar]}$, where $\cent$ ($\dollar$) means the automaton reads the left (resp., the right) end-marker and $n$ means that the class is defined by the automata with $\le n$ states.

Unless otherwise specified, all unary languages are defined on $ \{a\} $. As usual, $^*$ and $^+$ stand for the Kleeny star and the positive iteration, respectively, $\overline{L}$ is the complement of $L$ and $ \varnothing $ is the empty language. We define $ \even= (aa)^* $ and $ \less_n = \{ a^i \mid i \leq n \} $.

\section{Cardinality of unary languages}
\label{sec:card-unary}

GFAs, PFAs, and QFAs define the same class \cite{Tur69,YS09C,YS11A}:
\begin{equation}
	\label{eq:class-equality}
	\mathsf{S} = \mathsf{PseudoS} = \mathsf{QAL} \text{ \ and \ } \mathsf{UnaryS} = \mathsf{UnaryPseudoS} = \mathsf{UnaryQAL}.
\end{equation}
Note that using end-markers does not change the classes. On the other hand, $ \sf MCL[\cent\dollar] $ and $ \sf UnaryMCL[\cent\dollar] $ are proper subsets of $\mathsf{S}$ and $ \mathsf{UnaryS} $, respectively, since they contain no finite languages except for the empty language \cite{BC01B}.

In his seminal paper \cite{Rab63}, Rabin showed that the cardinality of $\mathsf{S}$ is uncountable by exhibiting a 2-state PFA on binary alphabet. To the best of our knowledge, a similar question for unary languages has been open up to now. In this section, we answer this question positively and provide the exact state bounds. We use rotations of the unit circle as transition matrices. Let $\theta \in [0,2\pi)$ be an angle. The \emph{rotation automaton} $\mathcal{R}_\theta$ is the 2-state GFA on the alphabet $\Sigma=\{a\}$ with the initial vector 
$ \bigl( \begin{smallmatrix} 1 \\ 0 \end{smallmatrix} \bigr) $, the transition matrix 
$ R_\theta = \bigl( \begin{smallmatrix} \cos \theta & - \sin \theta  \\ \sin \theta & \cos \theta \end{smallmatrix} \bigr)$ of the operator of the counter-clockwise rotation of the complex plane by the angle $\theta$, and the final vector $ ( \begin{smallmatrix} 1& 0 \end{smallmatrix}) $. The accepting value of $\mathcal{R}_\theta$ on the input $a^k$ ($k \geq 0$) is then equal to $ \cos (k \theta) $. Note the following simple fact.


\begin{fact} \label{dense}
If $\alpha$ is an irrational number, then the sequence of accepting values of the rotation automaton $R_{\alpha\pi}$ for the words $a^k$ is aperiodic and dense in $[-1,1]$.
\end{fact}

Now we pick the matrix $ R_{\theta} = \begin{pmatrix} 	3/5 & -4/5 \\ 4/5 & 3/5 \end{pmatrix} $ and consider the corresponding rotation automaton $\mathcal{R}_{\theta}$. By Fact~\ref{dense}, for any given $ \lambda_1 < \lambda_2 \in [0,1) $ there is an integer $ k > 0 $ such that 
$
	\lambda_1 < \cos (k \theta_1) < \lambda_2.
$
Therefore, we can follow that
\[
	L(\mathcal{R}_{\theta},\lambda_2) \subsetneq L(\mathcal{R}_{\theta},\lambda_1)
\]
since $a^k\in L(\mathcal{R}_{\theta},\lambda_1) \backslash L(\mathcal{R}_{\theta},\lambda_2)$. That is, for any given $\lambda \in [0,1) $, we obtain a different language $ L(\mathcal{R}_{\theta},\lambda) $. Thus, we have proved

\begin{theorem} \label{ups2}
	The cardinality of $ \sf UnaryPseudoS[2] $ is uncountable. 
\end{theorem}

\begin{remark}
Due to the aperiodicity of the sequence $f_{\cal M}(w)$, each $L(\mathcal{R}_{\theta},\lambda)$ is nonregular. 
\end{remark}

By \eqref{eq:class-equality}, $ \sf UnaryS$ and $\sf UnaryQAL $ also have uncountable cardinality. Moreover, the automaton $ \mathcal{R}_{\theta} $ is also a MCQFA with the accepting probability $ \cos^2 (k \theta) $ on the input $a^k$. So, for any given $ \lambda_1 < \lambda_2 \in [0,1) $, there is  some $ k > 0 $ such that $\lambda^2_1 < \cos^2 (k \theta) < \lambda_2^2  $. Repeating the rest of the proof of Theorem~\ref{ups2}, we get 

\begin{theorem} \label{umcl2}
	The cardinality of $ \sf UnaryMCL[2] $ (and hence of $ \sf UnaryMCL $ and of $ \sf UnaryQAL[2] $) is uncountable. 
\end{theorem}

The classes $\sf S$ and $\sf QAL$ remain the same when the cutpoint is fixed to a value between 0 and 1. But, this is not true for cutpoint 0. With the cutpoint 0, PFAs recognize only regular languages \cite{Paz71} and QFAs recognize ``exclusive'' stochastic languages ($\sf S^{\neq}$) but not all stochastic languages \cite{YS10A}. Note that unary ``exclusive'' stochastic languages are regular \cite{SS78}. 

It was an open question whether with cutpoint 0 MCQFAs recognize a proper subset of $\sf MCL $ \cite{YS10A}. Now we  answer this question in the affirmative. All unary languages recognized by MCQFAs with cutpoint 0 are regular as mentioned above, while $\sf UnaryMCL $ contains uncountably many unary nonregular languages.

\subsection{Small unary PFAs} \label{sec:basic-facts}

We continue with unary PFAs with few states. Contrary to GFAs and QFAs, 2-state unary PFAs recognize only regular languages. This fact was mentioned in \cite[Ch.~3]{Paz71} as Exercise~15. For the sake of completeness, we prove this result as Theorem~\ref{thm:paz-2-state-pfa}. Our proof explicitly lists all these regular languages. Another deep distinction of PFAs is the following. A single unary GFA or QFA  can define uncountable many languages by selecting different cutpoints. On the other hand, a unary $n$-state PFA defines at most $n$ nonregular languages, and hence, countably many languages at all \cite[Ch.~3, Ex.~11]{Paz71}. Thus, in order to prove that the cardinality of $\sf UnaryS[n]$ is uncountable for some $n$, we need a different argument.

It is known that 3-state unary PFAs recognize some nonregular languages \cite[Thm.~3.6]{Paz71}. The idea behind the proof of this statement can be developed to show the main result of this section (Theorem~\ref{thm:3-state-pfa}): the cardinality of $\sf UnaryS[3]$ is uncountable. A weaker result, namely, the fact that the cardinality of $\sf UnaryS[4] $ is uncountable, was proved in \cite{ShY14} using a quite different technique based on Turakainen's theorem  \cite{Tur75} about the ``conversion'' of GFAs into PFAs. We also note that both Theorems~\ref{thm:paz-2-state-pfa} and~\ref{thm:3-state-pfa} are proved in the strong form with respect to endmarkers: they are on in Theorem~\ref{thm:paz-2-state-pfa} and off in Theorem~\ref{thm:3-state-pfa}.

\begin{theorem} 
	\label{thm:paz-2-state-pfa} 
For any 2-state unary PFA $\cal P$ with endmarkers and any $\lambda \in [0,1)$, the language	$L(\mathcal{P},\lambda)$ is regular. 
\end{theorem}

\begin{proof}
Let $\mathcal{P} = \big(\{q_1,q_2\},\{a\},\{A_a = A\},v_0=(v_{01},v_{02})^\top,f=(f_1~~f_2)\big)$. The matrix $A$ can be written as
\[
	A=\left( \begin{array}{ccc} 1-x & ~~ & y \\ x & & 1-y \end{array}	 \right).
\]
If $x=y=0$, then $A$ is identity and for any input $a^m$ ($m\geq 0$) $f_{\cal P}(a^m) = f v_0$ is fixed. Then, $ L({\cal P},\lambda) $ is either $\varnothing$ or $a^*$.

If $x=y=1$, then $\cal P$ alternates between two probabilistic states:
\[
	v_0 = 
	\left(\begin{array}{cc} v_{01} \\ v_{02} \end{array}	 \right)
	,
	v_1 = \left(\begin{array}{cc} v_{02} \\ v_{01} \end{array}	 \right)
	,
	v_2 = v_0 = \left(\begin{array}{cc} v_{01} \\ v_{02} \end{array}	 \right)
	,
	v_3 = v_1 =  \left(\begin{array}{cc} v_{02} \\ v_{01} \end{array}	 \right)
	, \cdots
\]
That is, for any $m \geq 0$, $f_{\cal P}(a^{2m}) = f v_0$ and $f_{\cal P}(a^{2m+1}) = f v_1$. Then, $ L({\cal P},\lambda) $ can be $\varnothing$, $ \even $,  $ \overline{\even} $, or $a^*$. 

In the remaining part, we assume that $ x+y \in (0,2) $. The stationary distribution of $A$ is
\[ 
	\left( \begin{array}{c} \dfrac{y}{x+y} \\ \\ \dfrac{x}{x+y}
\end{array}	 \right).
\]
Since $v_0$ is stochastic, 
\[ 
	v_0 = \left( \begin{array}{c} \dfrac{y}{x+y} + c \\ \\ \dfrac{x}{x+y} -c
\end{array}	 \right)
\]
for some $ c \in \mathbb{R} $. After reading an $a$, the new state is
\[ 
	v_1 = Av_0 =\left( \begin{array}{c} \dfrac{y}{x+y} + c (1-(x+y)) \\ \\ \dfrac{x}{x+y} - c (1-(x+y))
\end{array}	 \right).
\]
So, the state after reading $j$ symbols is
\[ 
	v_j = \left( \begin{array}{c} \dfrac{y}{x+y} + c (1-(x+y))^j \\ \\ \dfrac{x}{x+y} - c (1-(x+y))^j
\end{array}	 \right).
\]
Then, the accepting probability of $\cal P$ on $a^m$ is 
\[
f_{\cal P}(a^m) = 
\begin{cases}
f_1  \dfrac{y}{x+y} + f_2  \dfrac{x}{x+y} + c (f_1-f_2) \left( 1-(x+y)\right)^m&\text{ if } m>0,\\
f_1  \dfrac{y}{x+y} + f_2  \dfrac{x}{x+y} + c (f_1-f_2)&\text{ if } m=0.
\end{cases}
\]
If $ f_1 = f_2  $, then the accepting probability is fixed for any string. Then, $ L({\cal P},\lambda) $ is either $\varnothing$ or $a^*$.

If $ x+y =1 $, then the accepting probability of any string of nonzero length is fixed, but, the accepting probability of the empty string $\varepsilon$ can be different. Then, $ L({\cal P},\lambda) $ can be $\varnothing$, $\{\varepsilon\}$, $ a^+ $, or $a^*$.

By also excluding these two cases, we can rewrite $ f_{\cal P}(a^m) $ as $ z+ r t^m $, where 
\[
	z=f_1  \dfrac{y}{x+y} + f_2  \dfrac{x}{x+y}, r=c (f_1-f_2), \text{ and } t=( 1-(x+y)).
\]
Since $|t|<1$, it is clear that $ f_{\cal P}(a^m) \rightarrow z  $ as $ m \rightarrow \infty$. If $t$ is positive, then $ f_{\cal P}(a^m) $ monotonely approaches  $z$. Thus, the possible values of $ L({\cal P},\lambda) $, depending on the cutpoint, are the languages $ \less_n $ for any  $ n \geq 0 $ and their complements $ \overline{ \less_n }$. If $t$  is negative, then  $ f_{\cal P}(a^m) $ shows a dying oscillation around $z$ with period 2. Thus, the language $ L({\cal P},\lambda) $ equals $  \less_n \cap \even$,  $ \less_n \cap \overline{ \even } $, $ \overline{ \less_n } \cap\even$, $\overline{\less_n} \cap \overline{ \even }$ ($ n\ge 0$), or the complement of one of these languages.
\qed\end{proof}

Remark that the absence of endmarkers does not change the class, i.e., $ \sf UnaryS[\cent 2 \dollar ] = UnaryS[2] $.

\begin{theorem} \label{thm:3-state-pfa}
	The cardinality of $\sf UnaryS[3] $ is uncountable.
\end{theorem}

\begin{proof}
For each $ x \in (0,\frac{1}{2}] $, we consider the stochastic matrix 
\[
	A_x = \left( \begin{array}{ccc}
		0~~  & 0  & x \\
		1~~ & 0 & x \\
		0~~ & 1 & 1 {-} 2x
\end{array}	 \right),
\]
and the corresponding PFA ${\cal P}_x=\big(\{q_1,q_2,q_3\},\{a\},A_x,(1~0~0)^\top,(0~0~1)\big)$. The eigenvalues of $  A_x $ are 
\[
	r_1=1,\ r_{2,3}=-x \pm \sqrt{x^2-x}.
\]
In the prescribed interval for $x$, two of them are complex numbers and can be written as 
\begin{equation} \label{eq_lambda}
r_{2,3}=\sqrt{x}(\cos\theta_x\pm i\sin\theta_x), \text{ where }\theta_x = \arccos(-\sqrt{x}) = \pi-\arcsin\sqrt{1{-}x}. 
\end{equation}
Let us fix $x$ and denote the entries of the matrix $A_x^m$ by $a_{ij}^{(m)}$. By the Cayley-Hamilton theorem,  $A_x$ satisfies its own characteristic equation. Then the sequence $\{A_x^m\}$ satisfies the linear homogeneous recurrence relation with the same characteristic equation. Therefore, this recurrence holds for any sequence $\{a_{ij}^{(m)}\}$. Since all roots of the characteristic equation are simple, $\{a_{ij}^{(m)}\}$ is a linear combination of sequences $\{r_1^m\},\{r_2^m\}$, and $\{r_3^m\}$ by the main theorem on linear recurrences. 

Note that $f_{{\cal P}_x}(a^m)=a_{31}^{(m)}$ by the definition of the automaton ${\cal P}_x$. Hence, 
\begin{equation} \label{eq_fPx}
f_{{\cal P}_x}(a^m)= A + br_2^m + cr_3^m,
\end{equation}
where the coefficients can be found from the initial conditions 
\begin{equation} \label{eq_012}
f_{{\cal P}_x}(a^0)=f_{{\cal P}_x}(a^1)=0,\ \ f_{{\cal P}_x}(a^2)=1.
\end{equation}
Since $r_2$ and $r_3$ are complex conjugates, $b$ and $c$ should be complex conjugates as well to make the sum \eqref{eq_fPx} a real number. To get rid of the complex-valued coefficients, we substitute $b=B+iC$ and $c=B-iC$ into \eqref{eq_fPx}. Taking \eqref{eq_lambda} into account, we obtain
\begin{multline} \label{eq_fPx2}
f_{{\cal P}_x}(a^m)= A + 2x^{m/2}(B\cos m\theta_x - C\sin m\theta_x)=\\
A + 2\sqrt{B^2+C^2}\cdot x^{m/2}\cdot\cos (m\theta_x +\gamma_x), \text{ where }\gamma_x=\arccos\frac{B}{\sqrt{B^2+C^2}}.
\end{multline}
The conditions \eqref{eq_012} give us a system of three linear equations in the variables $A,B$, and $C$. Solving this system, we obtain
$$
A=\frac1{3x+1},\ \ B=-\frac1{6x+2},\ \ C=\frac{x+1}{(6x+2)\sqrt{x-x^2}},
$$
and finally transform \eqref{eq_fPx2} into 
\begin{equation} \label{eq_fPx3}
f_{{\cal P}_x}(a^m)= \lambda_x + Dx^{m/2}\cos(m\theta_x+\gamma_x), 
\end{equation}
where $\lambda_x=\frac1{3x+1}\in(0,1)$, $D=\frac1{\sqrt{(3x+1)(x-x^2)}}>0$ and $\gamma_x=\arccos\big(-\sqrt{\frac{x-x^2}{3x+1}}\big)$.

\bigskip
To prove the theorem, it suffices to show that all languages of the form $L({\cal P}_x,\lambda_x)$ are distinct. By \eqref{eq_fPx3}, $a^m\in L({\cal P}_x,\lambda_x)$ if and only if $\cos(m\theta_x+\gamma_x)>0$. By \eqref{eq_lambda}, $x_1<x_2$ implies $ \theta_{x_1}<\theta_{x_2} $, and the set of all possible values of angles $\theta_x$ is the interval $\big(\frac{\pi}2,\frac{3\pi}4\big]$.

Let us fix $x_1,x_2\in (0,\frac12]$ satisfying $x_1<x_2$ and find $m$ such that 
$$
m(\theta_{x_2}-\theta_{x_1})+\gamma_{x_2}-\gamma_{x_1}\le\pi \text{ and }(m+1)(\theta_{x_2}-\theta_{x_1})+\gamma_{x_2}-\gamma_{x_1}>\pi.
$$
We partition $\mathbb{R}$ into the intervals of length $\pi$, in which the function $\cos\alpha$ does not change sign; all borderline points are attached to ``negative'' intervals:

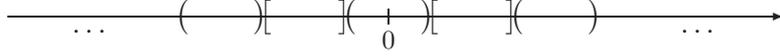
\centerline{
\unitlength=1mm
\begin{picture}(100,12)(0,-5)
\put(50,-2){\makebox(0,0)[ct]{\small0}}
\put(0,0){\vector(1,0){102}}
\put(50,-1){\line(0,1){2}}
\multiput(34,0)(22,0){3}{\makebox(0,0)[rc]{\big)}}
\multiput(22,0)(22,0){3}{\makebox(0,0)[lc]{\big(}}
\multiput(33,0)(22,0){2}{\makebox(0,0)[lc]{\big[}}
\multiput(45,0)(22,0){2}{\makebox(0,0)[rc]{\big]}}
\multiput(11,-2)(80,0){2}{\makebox(0,0)[cc]{$\cdots$}}
\end{picture}
}
\noindent By the choice of $m$, the numbers $\alpha_1=m\theta_{x_1}+\gamma_{x_1}$ and $\alpha_2=m\theta_{x_2}+\gamma_{x_2}$ differ by at most $\pi$ and then either both are borderline, or belong to the same interval, or belong to adjacent intervals. In the latter case, exactly one of the numbers $\cos\alpha_1$ and $\cos\alpha_2$ is positive; hence, the languages $L({\cal P}_{x_1},\lambda_{x_1})$ and $L({\cal P}_{x_2},\lambda_{x_2})$ are different, because exactly one of them contains the word $a^m$. In the former two cases, consider the numbers $\alpha_1'=(m+1)\theta_{x_1}+\gamma_{x_1}$ and $\alpha_2'=(m+1)\theta_{x_2}+\gamma_{x_2}$. If $\alpha_1$ and $\alpha_2$ are borderline points, then $\alpha_1'$ and $\alpha_2'$ belong to adjacent intervals following these points (recall that $\theta_{x_1}<\theta_{x_2}<\pi$). If $\alpha_1$ and $\alpha_2$ belong to the same interval, then each of $\alpha_1'$ and $\alpha_2'$ belongs to the same or the next interval. Since the distance between  $\alpha_1'$ and $\alpha_2'$ exceeds $\pi$ by the choice of $m$, they cannot belong to the same interval; so they belong to adjacent intervals. Similar to the above, we see that exactly one of the languages $L({\cal P}_{x_1},\lambda_{x_1})$ and $L({\cal P}_{x_2},\lambda_{x_2})$ contains $a^{m+1}$. The theorem is proved.
\end{proof}

\section{One-state pseudo stochastic languages}
\label{sec:one-state-pseudo}

In the previous section, we have shown that 2-state GFAs and QFAs can define uncountable many languages. So, it is interesting to consider the 1-state case. But 1-state QFAs (and so PFAs) are trivial. Indeed, they are always in the same state with probability 1 and so all strings have the same accepting probability. On the other hand, 1-state GFAs recognize many nontrivial languages. For example, the GFA $(\{q\},\{a,b\},\{A_a=(\frac{1}{2}),A_b=(2)\},v_0=1,f=1)$ recognizes the language of all words containing more $b$'s than $a$'s with cutpoint 1.

In this section, we completely describe the languages contained in $ \sf PseudoS[1] $ and relate them to regular and context-free languages. As a corollary, we get a characterization of $ \sf UnaryPseudoS[1] $. For convenience, we write $ \sf PseudoS[1,\Sigma] $ if the alphabet $\Sigma$ is fixed.

Suppose that $\Sigma=\{a_1,\ldots,a_n\}$, $w\in\Sigma^*$, and $|w|_{a_i}$ stands for the number of occurrences of the letter $a_i$ in $w$. Then $\pi(w)=(|w|_{a_1},\ldots,|w|_{a_n})$ is the \emph{Parikh vector} of $w$. Two words with equal Parikh vectors are \emph{anagrams}: they can be obtained from each other by resorting their letters. For a language $L$, $\pi(L)=\{\pi(w)\mid w\in L\}$ is the \emph{Parikh set} of $L$. A language $L$ is \emph{Parikh closed} if it contains all anagrams of any of its words. Parikh vectors appear in many studies on formal languages; a cornerstone result by Parikh \cite{Par66} says that for any context-free language has the same Parikh set as some regular language.

Let us introduce three types of Parikh closed languages. For arbitrary $\alpha\in\mathbb R\cup\{+\infty\}$, $b_1,\ldots,b_n\in\mathbb R$, the \emph{solution language} $\Sol(\Sigma, b_1,\ldots, b_n,\alpha)$ is the language whose Parikh set coincides with the set of all nonnegative integer solutions to the linear inequality $(\vec b, \vec  x)=b_1x_1+\cdots +b_nx_n<\alpha$. The numbers $b_1,\ldots, b_n$ are \emph{coefficients} of the language. For a given $Y\subseteq\Sigma$, the \emph{parity language} $\Par(\Sigma,Y,0)$ [resp., $\Par(\Sigma,Y,1)$] consists of all words from $\Sigma^*$ having even [resp., odd] number of occurrences of letters from $Y$. Finally, the \emph{indicator language} $\Ind(\Sigma,Y)$ consists of all words containing at least one letter from $Y$. In particular, one has $\Par(\Sigma,\varnothing,0)=\Sigma^*$, $\Par(\Sigma,\varnothing,1)=\Ind(\Sigma,\varnothing)=\varnothing$. By convention, we put 
$$
\Sol(\varnothing,\alpha)=
\begin{cases}
\{\varepsilon\}&\text{ if }\alpha>0,\\
\varnothing&\text{ if }\alpha\le 0.
\end{cases}
$$
It is easy to see that all parity languages and indicator languages are regular. On the other hand, most of the solution languages are not regular. For example, the inequality $x_1-x_2<0$ generates the above mentioned binary language $\{w \in \{a,b\}^* \mid |w|_a < |w|_b \} $.

\begin{theorem} \label{psS1}
For a fixed finite alphabet $\Sigma$, let $\bf\Lambda$ be the set of all languages of the form 
\begin{equation} \label{vee}
\Sol(X, b_1,\ldots, b_{|X|},\alpha)\cap\Par(X,Y,i),
\end{equation}
where $Y\subseteq X\subseteq\Sigma$, $i\in\{0,1\}$. Further, let $\bf V$ be set of all languages of the form 
\begin{equation} \label{wedge}
\Sol(X, b_1,\ldots, b_{|X|},\alpha)\cup\Par(X,Y,i)\cup\Ind(\Sigma,\Sigma\backslash X),
\end{equation}
where $Y\subseteq X\subseteq\Sigma$, $i\in\{0,1\}$, $\alpha\ne+\infty$. Then
\begin{equation} \label{ps}
\sf PseudoS[1,\Sigma]=\bf\Lambda\cup\bf V.
\end{equation}
\end{theorem}

\begin{proof}
The $1\times 1$ matrices are just real numbers, so we replace ``transition matrices'' with ``transition numbers'' in our terminology.
The multiplication of transition numbers is commutative, and this fact has two consequences. First, any $L\in \sf PseudoS[1]$ is Parikh closed. Second, the individual values of $v_0$, $f$, and $\lambda$ do not matter; namely, one can put $\lambda'=\frac\lambda{fv_0}$ and consider two possible acceptance conditions\footnote{A GFA with $v_0=0$ or $f=0$ recognizes either $\varnothing$ or $\Sigma^*$. The same effect can be achieved by setting all transition numbers to 0. Hence we assume w.l.o.g. $v_0,f \ne 0$. We also use the standard convention that $0^0=1$.}: 
\begin{equation}\label{acccon1}
A_{a_1}^{|w|_{a_1}}A_{a_2}^{|w|_{a_2}}\cdots A_{a_n}^{|w|_{a_n}} < \lambda' \text{ and } A_{a_1}^{|w|_{a_1}}A_{a_2}^{|w|_{a_2}}\cdots A_{a_n}^{|w|_{a_n}} > \lambda'.
\end{equation}
So, below we assume that a 1-state GFA over an $n$-letter alphabet $\Sigma$ is given by an $n$-tuple $\vec{A}=(A_1=A_{a_1},\ldots,A_n=A_{a_n})$ of real numbers. The cutpoint $\lambda=\lambda'$ and an additional bit to choose among the conditions \eqref{acccon1} are given separately.

We say that a 1-state GFA $\cal G$ is \emph{positive} if all numbers $A_i$ and $\lambda$ are positive. If $\pi(w)=(x_1,\ldots,x_n)$, then the acceptance condition
\begin{equation}\label{acccon2}
A_1^{x_1}\cdots A_n^{x_n} < (>) \lambda
\end{equation}
for a positive GFA can be rewritten as 
\begin{equation}\label{acccon3}
x_1\log A_1+\cdots+x_n\log A_n < (>) \log\lambda,
\end{equation}
where the logarithms are taken at any base greater than 1. But this linear inequality defines either the language $\Sol(\Sigma, \log A_1,\ldots,\log A_n, \log\lambda)$ (for the ``$<$'' sign in \eqref{acccon3}) or the language $\Sol(\Sigma, -\log A_1,\ldots,-\log A_n, -\log\lambda)$ (for the ``$>$'' sign)\footnote{From the geometric point of view, a 1-state positive GFA defines a hyperplane in $\mathbb R^n$ and accepts exactly the words having the ends of their Parikh vectors on the prescribed side of this hyperplane.}.

\smallskip
Now we proceed with the general case. We assume the ``$<$'' sign in \eqref{acccon2}; the case of the ``$>$'' sign admits a completely similar proof, so we omit it. For convenience, we reorder the alphabet such that the numbers $A_1,\ldots,A_k$ are nonzero, while the other transition numbers, if any, are zero. We also put $X=\{a_1,\ldots,a_k\}$ and denote the set of letters with negative transition numbers by $Y$. There are two possibilities. If $\lambda\le 0$, the inequality \eqref{acccon2} for the Parikh vector $(x_1,\ldots,x_n)$ of a word $w$ is equivalent to the conjunction of the following conditions:
\begin{itemize}
\item $w$ contains no letters from outside $X$;
\item the number of letters from $Y$ in $w$ is odd;
\item $|A_1|^{x_1}\cdots |A_k|^{x_k} > |\lambda|$.
\end{itemize}
The first two conditions define the language $\Par(X,Y,1)$, and the first and the third conditions define $\Sol(X, -\log|A_1|,\ldots,-\log|A_k|, -\log|\lambda|)$ (assuming $\log 0=-\infty$). Thus, we get a language from $\bf\Lambda$.

The second possibility is $\lambda>0$. Here \eqref{acccon2} is equivalent to the disjunction of the conditions
\begin{itemize}
\item $w$ contains a letter from outside $X$;
\item the number of letters from $Y$ in $w$ is odd;
\item $|A_1|^{x_1}\cdots |A_k|^{x_k} < \lambda$.
\end{itemize}
Similar to the above, these conditions define a language in $\bf V$ (note that $\alpha$ is finite because $\lambda>0$). Hence we obtain $\sf PseudoS[1,\Sigma]\subseteq\bf\Lambda\cup\bf V$. 

\smallskip
In order to show the reverse inclusion, we use the above considerations to build 1-state GFA's with appropriate acceptance conditions from the elements of $\bf\Lambda\cup\bf V$. Let us first take a language $\Sol(X, b_1,\ldots, b_k,\alpha)\cap\Par(X,Y,i)$ (as above, we assume $X=\{a_1,\ldots,a_k\}$). We put 
\begin{equation} \label{Aj}
A_j=\begin{cases}
0,& \text{if } j>k,\\
2^{b_j},& \text{if } a_j\in X\backslash Y,\\
-2^{b_j},& \text{if } a_j\in Y.
\end{cases}
\end{equation}
If $i=1$, then we use the acceptance condition ``${<}{-2^{-\alpha}}$''. For $i=0$, the condition is ``${>}{2^{-\alpha}}$''. In the case of a language $\Sol(X, b_1,\ldots, b_k,\alpha)\cup\Par(X,Y,i)\cup\Ind(\Sigma,\Sigma\backslash X)$ we also use \eqref{Aj} to define a GFA, but the acceptance conditions are different: ``${>}{2^{\alpha}}$'' for $i=1$ and ``${<}{-2^{\alpha}}$'' for $i=0$. Thus we have $\sf PseudoS[1,\Sigma]\supseteq\bf\Lambda\cup\bf V$. The theorem is proved.
\end{proof}

\begin{corollary}
The cardinality of ${\sf UnaryPseudoS[1},\{a,b\}]$ is uncountable.
\end{corollary}

\begin{proof}
On a plane, there are uncountably many pairwise non-parallel lines, and all of them define different solution languages.
\end{proof}

\begin{corollary}
All languages in $\sf UnaryPseudoS[1]$ are regular. Moreover, the class 
$\sf UnaryPseudoS[1]$ consists of 
\begin{equation}\label{ups1}
\varnothing,\ \even,\ \less_n,\ 
\less_n \cup \even,\  \less_n \cup \overline{ \even },\ \overline{ \less_n } \cup\even,\ \overline{\less_n} \cup \overline{ \even }, 
\end{equation}
where $n\ge 0$, and the complements of the languages \eqref{ups1}.
\end{corollary}

\begin{proof}
The possible cases are $X=\varnothing$, $X=Y=\{a\}$, and $(X=\{a\}\wedge Y=\varnothing)$. We have
$\Sol(\{a\},b,\alpha)\in\{\less_n,\overline{\less_n}\}$ for some $n$, $\Par(\{a\},\{a\},0)=\even$, $\Par(\{a\},\{a\},1)=\overline{\even}$, $\Ind(\{a\},\{a\})=a^+$. Now the required list of languages can be obtained directly from \eqref{vee}--\eqref{ps}.
\end{proof}

The obtained list coincides with the one from the proof of Theorem~\ref{thm:paz-2-state-pfa}: 
\begin{corollary} \label{ps1s2}
$ \sf UnaryPseudoS[1] = UnaryS[2]$. 
\end{corollary}
Now we are going  to relate the 1-state pseudo stochastic languages to the classes of the Chomsky hierarchy. We need some additional notions. 
Real numbers $b_1,\ldots, b_n$ are called \emph{rationally equivalent} if there exist $\gamma\in\mathbb R$, $q_1,\ldots, q_n\in\mathbb Q$ such that $b_i=q_i\gamma$ for all $i$. Let $L=\Sol(\Sigma, b_1,\ldots, b_n,\alpha)$, $N\subseteq\Sigma$ be the set of letters corresponding to zero coefficients $b_i$. By \emph{decimation} $\dec(L)$ of $L$ we mean the language over $\Sigma\backslash N$ obtained from $L$ by deleting all letters of $N$ from all words (if $N=\varnothing$, then $\dec(L)=L$). 

\begin{lemma} \label{decim}
A solution language $L$ is regular if and only if $\dec(L)$ is regular. 
\end{lemma}

\begin{proof}
If $L$ is regular, one can take its recognizing DFA and replace all labels from $N$ by $\varepsilon$, getting a $\varepsilon$-NFA recognizing $\dec(L)$. For the converse, note that
$$
L=\bigcup_{w=c_1\cdots c_k\in\dec(L)}\!\!\! N^*c_1N^*c_2\cdots N^*c_kN^*.
$$
So, one can take a DFA recognizing $\dec(L)$ and add loops labeled by all letters from $N$ to each its state. The resulting automaton will recognize $L$. 
\end{proof}

The following lemma is crucial.

\begin{lemma} \label{solregcf}
1) A solution language is regular if and only if all its nonzero coefficients have the same sign.\\ 
2) A nonregular solution language is context-free if and only if its coefficients are rationally equivalent.
\end{lemma}

\begin{proof}
First we note that the decimation of a solution language is a solution language defined by the same linear inequality in a vector space of a smaller dimension. Let $D$ be the decimation of some solution language. The Parikh vectors of its words satisfy an inequality $b_1x_1+\cdots b_kx_k<\alpha$, where all coefficients $b_i$ are nonzero. If all these coefficients are positive [resp. negative], then $D$ is finite [resp., cofinite]. By Lemma~\ref{decim}, in this case any solution language with the decimation $D$ is regular. Now assume that some coefficients have different signs; w.l.o.g., $b_1>0>b_2$. Let the letters $a_1$ and $a_2$ correspond to $b_1$ and $b_2$, respectively. If $D$ is regular, then it is recognized by a DFA $\cal A$ with, say, $t$ states. This DFA accepts all words from $D$ including all words $a_1^{x_1}a_2^{x_2}$ such that $b_1x_1+b_2x_2<\alpha$. Such words exist for any $x_1$, in particular, for $x_1>t$. Then $\cal A$ has a cycle labeled by some $a_1^i$, $i\le t$. Iterating this cycle appropriate number of times, we will get a word of the form $a_1^{x_1+ri}a_2^{x_2}$ which is recognized by $\cal A$ but does not belong to $D$. Thus, $D$ is not regular, and a reference to Lemma~\ref{decim} finishes the proof of statement 1.

\smallskip
Now we turn to the proof of statement 2. Take a solution language $L$ with the decimation $D=\Sol(\Sigma, b_1,\ldots, b_k,\alpha)$. Since $L$ is not regular, we know from the above that some $b_i$'s have different signs; w.l.o.g., $b_1>0>b_2$.  

Both $L$ and $D$ are determined by the inequality $b_1x_1+\cdots b_kx_k< \alpha $.
If the coefficients are rationally equivalent, we transform this inequality, dividing both sides by the common irrational factor of all coefficients and than multiplying both sides by the least common multiple of denominators of the obtained rational coefficients. As a result, we get a linear inequality
$$
\hat b_1x_1+\cdots \hat b_kx_k< \hat\alpha
$$
with integer coefficients and the same set of solutions. Finally, we replace $\hat\alpha$ by $\lceil\hat\alpha\rceil$ preserving the set of \emph{integer} solutions of the inequality. To check whether the Parikh vector of a word satisfies the resulting diophantine inequality, one can implement a counter in the stack of a pushdown automaton. Hence, the solution languages with rationally equivalent coefficients are context-free.

Now consider a solution language $L$ having rationally non-equivalent coefficients. If any positive coefficient is equivalent to any negative one, then all coefficients are equivalent; so, $L$ has a pair of rationally non-equivalent coefficients of different signs, say, $b_1$ and $b_2$. Then the value of the expression $b_1x_1+b_2x_2$ for the word $a_1^{x_1}a_2^{x_2}\in L$ can be arbitrarily close from below to $\alpha$. Thus, $(\vec b,\pi(w))$ for $w\in L$ can be arbitrarily close from below to $\alpha$ (and the supremum cannot be reached by the definition of solution language). Let us show that this is impossible for context-free languages. Aiming at a contradiction, assume that $L$ is context-free. By Parikh's Theorem \cite{Par66} there exists a regular language $L'$ such that $\pi(L')=\pi(L)$. Since $L$ is infinite, $\pi(L)$ and $L'$ are infinite as well. Consider the minimal DFA $\cal A$ with partial transition function, recognizing $L'$. This DFA must contain cycles; let $z$ be the label of some cyclic walk in the graph of $\cal A$. Then for some $u,v\in\Sigma^*$ the language $L'$ contains the words $uz^tv$ for all nonnegative integers $t$. Hence we have
$$
(\vec b,\pi(uz^tv))=(\vec b,\pi(uv))+t(\vec b,\pi(z))<\alpha,
$$
implying $(\vec b,\pi(z))\le 0$. Since this inequality holds for the label of any cyclic walk, the function $(\vec b,\pi(w))$ reaches its maximum for $w\in L'$ on some short word $w$. Thus, the maximum of $(\vec b,\pi(w))$ for $w\in L$ is also reachable, a contradiction. Hence, $L$ is not context-free.
\end{proof}

Now we are able to relate $\sf PseudoS[1]$ to the classes of the Chomsky hierarchy.

\begin{theorem}
1) A 1-state pseudo stochastic language is regular if and only if the logarithms of absolute values of all nonzero transition numbers of the generating 1-state GFA have the same sign.\\ 
2) A nonregular 1-state pseudo stochastic language is context-free if and only if the logarithms of absolute values of all nonzero transition numbers of the generating 1-state GFA are rationally equivalent.
\end{theorem}

\begin{remark}
It is easy to check that the properties ``to have the same sign'' and ``to be rationally equivalent'' for logarithms are independent of the base of the logarithm.
\end{remark}

\begin{proof}
By \eqref{ps}, a language $L\in \sf PseudoS[1]$ is given either by \eqref{vee} or by \eqref{wedge}. In both cases, $L$ is regular [context-free] if and only if the corresponding solution language is regular [resp., context-free]. As was shown in the proof of Theorem~\ref{psS1}, the coefficients of this solution language are logarithms of absolute values of the transition numbers of the GFA recognizing $L$. The result now follows from Lemma~\ref{solregcf}.
\end{proof}

\begin{remark}
From the proof of Lemma~\ref{solregcf} one can conclude that if a 1-state pseudo stochastic language is context-free, it is deterministic context-free.
\end{remark}

\section{Inclusive and exclusive cutpoint languages}
\label{sec:incl-excl}

For a given automaton $ \cal M $ and a cutpoint $ \lambda \in [0,1] $, the languages $ L(\mathcal{M},{=}\lambda) $ and $ L(\mathcal{M},{\neq}\lambda) $ are  defined by 
\[
	L(\mathcal{M},{=}\lambda) = \{ w \in \Sigma^* \mid f_{\mathcal{M}}(w) = \lambda \}
	,\
	L(\mathcal{M},{\neq}\lambda) = \{ w \in \Sigma^* \mid f_{\mathcal{M}}(w) \neq \lambda \},
\]
where $ \lambda \in \mathbb{R} $ for GFAs and $ \lambda \in [0,1] $ for PFAs and QFAs. The language $ L(\mathcal{M},{=}\lambda) $ [resp., $L(\mathcal{M},{\neq}\lambda)$] is said to be recognized by $ \cal M $ with inclusive [resp., exclusive] cutpoint $ \lambda $. (Note that if a language is recognized by an automaton with inclusive cutpoint $\lambda$, then its complement is recognized by the same automaton with exclusive cutpoint $\lambda$.) Such languages recognized by GFAs [PFAs, QFAs] are called inclusive and exclusive pseudo stochastic [resp., stochastic, quantum automaton] languages. The corresponding class names are given below:
\begin{equation*}
	\begin{array}{|l|l|l|l|l|}
		\hline 
		\multicolumn{1}{|c|}{ \mathbf{model} } & \multicolumn{1}{r}{\mathbf{general} } & \multicolumn{1}{l}{\mathbf{alphabet}} & \multicolumn{1}{|r}{ \mathbf{unary} } & \multicolumn{1}{l|}{ \mathbf{alphabet} }
		\\ \hline \hline
		\mbox{GFA} & \sf PseudoS^{=} & \sf PseudoS^{\neq} & \sf UnaryPseudoS^{=} & \sf UnaryPseudoS^{\neq} 
		\\ \hline
		\mbox{PFA} & \sf S^{=} & \sf S^{\neq} & \sf UnaryS^{=} & \sf UnaryS^{\neq}
		\\ \hline
		\mbox{QFA} & \sf QAL^{=} & \sf QAL^{\neq} & \sf UnaryQAL^{=} & \sf UnaryQAL^{\neq}
		\\ \hline
		\mbox{MCQFA} & \sf MCL^{=} & \sf MCL^{\neq} & \sf UnaryMCL^{=} & \sf UnaryMCL^{\neq}
		\\ \hline
	\end{array}
\end{equation*}

It is known that GFAs, PFAs, and QFAs define the same class of languages with inclusive and exclusive cutpoints \cite{YS10A}
\begin{equation}\label{eqneq}
	\sf  PseudoS^{=} = S^{=} = QAL^{=} \mbox{ and } PseudoS^{\neq} = S^{\neq} = QAL^{\neq} ,
\end{equation}
where inclusive and exclusive cutpoint languages form different classes \cite{Paz71} and we still do not know whether their intersection, which includes all regular languages, contains a non-regular language. On the unary alphabet, both classes coincide with regular languages; see the proof of Theorem 5.1 in \cite{SS78}.  

In Sect.~\ref{sec:ie} we find the cardinality of the classes \eqref{eqneq} (which is the same, because complementation is a bijection between the inclusive and exclusive classes), thus solving an open problem stated in \cite{YS10A}, and then relate $\sf MCL^{=}$ and $\sf MCL^{\neq}$ to these classes. Then in Sect.~\ref{sec:ie-few} we analyze inclusive and exclusive languages having few states.

\subsection{Two problems on the classes of inclusive and exclusive languages} \label{sec:ie}

What happens if we fix the cutpoint to a specific value? The classes $ \sf PseudoS^{=} $ and $ \sf PseudoS^{\neq} $ remain the same if we require the cutpoint to be any fixed real number; the same result holds for  $ \sf QAL^{=} $, $ \sf QAL^{\neq} $ and any cutpoint inside $ [0,1] $ \cite{YS11A}. On the other hand, although $ \sf S^{=} $ and  $ \sf S^{\ne} $ do not change when a cutpoint from $ (0,1) $ is fixed, the choice of 0 or 1 as a cutpoint shrinks each of these classes to the class of regular languages \cite{Bu67,Tur68,Ma93}. Note that for PFAs and QFAs the cutpoint 0 is equivalent to the exclusive cutpoint 0. PFAs with cutpoint 0 are equivalent to nondeterministic finite automata. Similarly, one can define nondeterministic quantum finite automata (NQFAs) as QFAs with cutpoint 0 \cite{YS10A}. So, the class of languages defined by NQFAs, named $ \sf NQAL $, is equivalent to $ \sf QAL^{\neq}  $ \cite{YS10A}. This connection lets us to prove

\begin{theorem}
	The cardinality of $ \sf QAL^{\neq} $ is countable. 
\end{theorem}

\begin{proof}
	It is clear that $ \sf NQAL (= QAL^{\neq}) $  is a subset of the class $ \sf NQP $ consisting of languages recognized by polynomial-time nondeterministic quantum Turing machines. In \cite{YY99}, it was shown that $ \sf NQP $ (defined with arbitrary complex numbers) is equivalent to  $ \sf coC_{=}P $,  the class of decision problems solvable by polynomial time nondeterministic Turing machines (NTMs) with the property that the number of accepting paths is different from the number of rejecting paths if and only if the answer is ``yes''. Since NTMs are ``classical'' Turing machines, their number is countable, as well as the cardinality of $ \sf coC_{=}P $.
\end{proof}


Now we relate the classes $ \sf MCL^{\neq} $ and $ \sf MCL^{=} $ to the classes \eqref{eqneq}.

\begin{theorem}
	\label{thm:mcqfa-cutpoint}
	Any language $ L \in \mathsf{MCL^{\neq}} $ can be defined by a MCQFA with exclusive cutpoint 0.
\end{theorem}
\begin{proof}
	Let $ \mathcal{M} = (Q,\Sigma,\{U_{\sigma} \mid \sigma \in \Sigma \},\ket{v_0},P) $ be a MCQFA with  $n$ states and the left end-marker, defining the language $ L $ with exclusive cutpoint $ \lambda $. If $\lambda=0$,  we are done, so let $\lambda\in (0,1] $. Thus, for an input $w \in \Sigma^*$, $ f_{\mathcal{M}}(w) \neq \lambda $ if $ w \in L $ and  $ f_{\mathcal{M}}(w) = \lambda $ if $ w \notin L $. Since the left end-marker is used, $ \ket{v_0} $ can be an arbitrary quantum state. If $m$ is the length of $w$, the quantum state of $ \mathcal{M} $ before the measurement is
	\[
		\ket{v_m} = U_{w_{m}}U_{w_{m-1}} \cdots U_{w_1} \ket{v_0} = \left( \begin{array}{c}
			\alpha_1 \\ \alpha_2 \\ \vdots \\ \alpha_n
\end{array}		 \right)
	\] 
	Applying the measurement, we obtain $ \ket{\tilde{v_a}} =P_a \ket{v_m} $. Note that $ \ket{\tilde{v_a}}  $ can be  obtained from $ \ket{v_m} $ by replacing certain entries with zeros. Namely, the $ j $th entry is replaced by 0 if the $(j,j)$th entry of $ P_a $ is 0 and preserved if the $ (j,j) $th entry of $ P_a $ is 1. Let $A \subseteq \{1,\ldots,m\} $ be the set of indices of the preserved entries. We refer to $ \{ q_j \in Q \mid j \in A \} $ as to the set of accepting states.
	
	For the accepting probability of $ w $ by $ \mathcal{M} $ one has 
	\[
f_{\mathcal{M}}(w) = 	\braket{\tilde{v_a}}{\tilde{v_a}}	= \sum_{j \in A} \alpha_j^*\alpha_j.
	\]

	We first construct an intermediate MCQFA $ \mathcal{M'} $ which executes two copies of $ \mathcal{M} $ in parallel. By definition, $ \mathcal{M'} = (Q',\Sigma,\{ U'_{\sigma} \mid \sigma \in \Sigma \},\ket{v'_0},P') $ is a tensor product of $ \mathcal{M} $ with itself:
	\begin{itemize}
		 \item $ Q' = Q \times Q $,
		 \item $ U_\sigma' = U_\sigma^* \otimes U_\sigma $,
		 \item $ \ket{v'_0} = \ket{v^*_0} \otimes \ket{v_0} $, and
		 \item $ P' $ is any measurement operator. 
	 \end{itemize} 
	 The quantum state of $\mathcal{M'}$ before the measurement is $ \ket{v'_m} = \ket{v^*_m} \otimes \ket{v_m} $, i.e.
	\[
\ket{v'_m} =  (U^*_{w_{m}} \otimes U_{w_{m}}) (U^*_{w_{m-1}} \otimes U_{w_{m-1}}) \cdots (U^*_{w_1} \otimes U_{w_1} )  (\ket{v^*_0} \otimes \ket{v_0}).
	\]
	Note that the entries of $ \ket{v'_m} $ form the set $ \{ \alpha^*_j\alpha_l \mid 1 \leq j,l \leq n \} $. 
	
	Now we define the target MCQFA $ \mathcal{M''} = (Q'',\Sigma,\{ U''_{\sigma} \mid \sigma \in \Sigma \cup \{\dollar\} \},\ket{v''_0},P'') $, which also uses the right end-marker:
	\begin{itemize}
		\item $ Q'' $ consists of $ Q' $ and one more state (the first one);
		\item the new initial state is 
		$
			\ket{v''_0} =
			\frac{1}{\sqrt{2}} \left(
				\begin{array}{c}
					1 
					\\				
					\ket{v'_0} 
				\end{array}
			\right)
	$ (its norm obviously equals 1);
	\item for each $ \sigma \in \Sigma $, 
	$
		U''_{\sigma} = 
		\left(
		\begin{array}{c|c}
				1 & 0 \cdots 0
				\\ \hline		
				0 &
				\\[-3pt]
				\vdots & U'_{\sigma}
				\\
				0 &
			\end{array}
		\right)
	$,
	and $ U_{\dollar} $ is described below; 
	\item $ P''_a $ has a single 1, located in the $ (1,1) $th entry. 
	\end{itemize}	
	A straightforward calculation shows that $ \ket{v''_m} = \frac{1}{\sqrt{2}} \left(
				\begin{array}{c}
					1 
					\\				
					\ket{v'_m} 
				\end{array}
			\right) $. 
By the definition of $ P''_a $, the accepting probability is the square of modulus of the first entry of $ U''_\dollar \ket{v''_m} $. Then, only the first row of $ U''_\dollar $ is essential, and so the remaining entries of this matrix can be arbitrary. We define the first row of $ U''_{\dollar} $ as 
	\[
		c\cdot ( - \lambda ~~ u ),
	\]
	where $ u  $ is an $ n^2 $-dimensional 0-1 row vector and the coefficient $ c $ sets the norm of the whole vector to 1. Here $ u $ is a kind of filtering such that its inner product with $ \ket{v'_m} = \ket{v^*_m} \otimes \ket{v_m} $ equals the sum
	\[
		\sum_{j \in A} \alpha^*_j \alpha_j = f_{\mathcal{M}}(w).
	\] 
	Thus, the first entry of $ U''_\dollar \ket{v''_m} $ is $ \frac{c}{\sqrt{2}} ( f_{\mathcal{M}} - \lambda ) $ and then the accepting probability of $ w $ by $ \mathcal{M''} $ is 
	\[
		 f_{\mathcal{M''}} (w) = \frac{c^2}{2} ( f_{\mathcal{M}} - \lambda )^2.
	\]
	If $ w \in L $, then the new accepting probability is nonzero, and, if $ w \notin L $, then it is zero. Therefore, $ L $ is defined by the MCQFA $ \cal M'' $ with exclusive curpoint 0. Remark that, as pointed before, any MCQFA with two end-markers is equivalent to a MCQFA with one end-marker. 
\end{proof}

\begin{corollary}
	\label{cor:MCL-classes}
	The class $ \sf MCL^{\neq} $ contains no non-empty finite languages.
\end{corollary}
\begin{proof}
	By Theorem~\ref{thm:mcqfa-cutpoint}, any language in $ \sf MCL^{\neq} $ is defined by a MCQFA with exclusive cutpoint 0, which is, in turn,  a MCQFA with cutpoint 0. But, as was mentioned in Sect.~\ref{sec:card-unary}, MCQFAs define no finite languages except for the empty one \cite{BC01B}. 
\end{proof}

Corollary~\ref{cor:MCL-classes} and complementation in \eqref{eqneq} immediately imply

\begin{corollary}
	The classes $ \sf MCL^{\neq} $ and $ \sf MCL^= $ are proper subsets of $ \sf QAL^{\neq} $ and $ \sf QAL^{=} $, respectively. The same relations hold in the unary case.
\end{corollary}

\subsection{Inclusive and exclusive languages with few states} \label{sec:ie-few}

Here we examine the inclusive and exclusive classes defined with very small number of states. We focus on the classes defined with inclusive cutpoint, since the classes for exclusive cutpoint can be then obtained by taking complements of languages.

We use the techniques from Sect.~\ref{sec:one-state-pseudo} in a straightforward way to characterize one-state inclusive pseudo stochastic languages.

For arbitrary $\alpha\in\mathbb R\cup\{+\infty\}$, $b_1,\ldots,b_n\in\mathbb R$, the \emph{equality solution language} $\Sol^{=}(\Sigma, b_1,\ldots, b_n,\alpha)$ is the Parikh-closed language whose Parikh set coincides with the set of all nonnegative integer solutions to the linear equation $(\vec b, \vec  x)=b_1x_1+\cdots +b_nx_n = \alpha$. Remark that equality solution languages can be non-regular; e.g., the equality $x_1-x_2 = 0$ generates the non-regular binary language $\EQ=\{w \in \{a,b\}^* \mid |w|_a = |w|_b \} $.

\begin{theorem} \label{psS1=}
For a fixed finite alphabet $\Sigma$, let $\bf\Lambda^=$ be the set of all languages of the form 
\begin{equation} \label{vee=}
\Sol^=(X, b_1,\ldots, b_{|X|},\alpha)\cap\Par(X,Y,i),
\end{equation}
where $Y\subseteq X\subseteq\Sigma$, $i\in\{0,1\}$. Then
\begin{equation} \label{ps=}
{\sf PseudoS}^=[1,\Sigma]={\bf\Lambda}^= \cup \{\Ind(\Sigma, X)\mid X\subseteq \Sigma\}.
\end{equation}
\end{theorem}

\begin{proof}
We adopt the proof of Theorem~\ref{psS1}. The analog of \eqref{acccon2}, representing the condition for accepting a word $w$ with the Parikh vector $\pi(w)=(x_1,\ldots,x_n)$ is 
\begin{equation}\label{acccon2=}
A_1^{x_1}\cdots A_n^{x_n} = \lambda .
\end{equation}
First, let $\lambda\ne0$ and let $X$ [resp., $Y$] denote the set of letters with nonzero [resp., negative] transition numbers. Then $w\in X^*$, and the accepted language is $\Sol^=(X, \log|A_1|,\ldots,\log|A_k|, \log|\lambda|)\cap \Par(X,Y,i)$, where $i=0$ [resp., $i=1$] for $\lambda>0$ [resp., $\lambda<0$]. Thus, any accepted language is given by \eqref{vee=}. Conversely, for any language \eqref{vee=} we apply \eqref{Aj} to build the corresponding GFA.

Now let $\lambda=0$. Then $w$ contains a letter with zero transition number and hence belong to an indicator language. The converse is also trivial, so we get \eqref{ps=}.
\end{proof}

\begin{corollary} The class $ \mathsf{UnaryPseudoS^{=}[1]} $ consists of the languages $ \varnothing $, $ a^* $, $ a^+ $, $ \even  $, $ \overline{\even} $, and $ \{a^n\} $ for $ n \geq 0 $. In particular, $ \mathsf{UnaryPseudoS^{=}[1]} \ne \mathsf{UnaryPseudoS^{\ne}[1]}$.
\end{corollary}

\begin{proof}
The unary equality solution language $\Sol^=(\{a\},b,\alpha)$ equals $\{a^n\}$ or $\varnothing$ if $b\ne0$; $a^*$ if $b=\alpha=0$;  $\varnothing$ if $b=0$, $\alpha\ne 0$. Since $\Par(\{a\},\{a\},0)=\even$, $\Par(\{a\},\{a\},1)=\overline{\even}$, $\Par(\{a\},\varnothing,0)=a^*$, $\Par(\{a\},\varnothing,1)=\Ind(\{a\},\varnothing)=\varnothing$, and $\Ind(\{a\},\{a\})=a^+$, we get the required list from \eqref{vee=}, \eqref{ps=}.
\end{proof}

The behaviour of 2-state PFAs on unary alphabet is examined in the proof of Theorem \ref{thm:paz-2-state-pfa}. Then, in a straightforward way, we can list all unary languages defined by these PFAs, arriving at the following analog of Corollary~\ref{ps1s2}.

\begin{corollary} 
$ \sf UnaryPseudoS^=[1] = UnaryS^=[2] $.
\end{corollary}

On the other hand, $\sf UnaryMCL^=[2]$ is incomparable with $ \sf UnaryPseudoS^=[1]$. Indeed, let $n>2$ be an integer. A 2-state MCQFA can start in state $ \ket{q_1} $, make a rotation with the angle $ \frac{\pi}{n} $ for each $a$ and accept the input if the state $ q_1 $ is observed \cite{AF98}. This MCQFA defines, with cutpoint 1, the language $ (a^n)^* $ which is not a member of $ \sf UnaryPseudoS^=[1] $. On the other hand, we know from Corollary~\ref{cor:MCL-classes} that no MCQFA can define $\{\varepsilon\}$ with an exclusive cutpoint, and hence $a^+\notin \sf UnaryMCL^=$. 

\smallskip
We close this section a couple of observations. While the class $\sf PseudoS^=$ is stable with respect to fixing the cutpoint to any particular number, its subclass $\sf PseudoS^=[1]$ ``discriminates'' the cutpoint 0: Theorem~\ref{psS1=} says that only indicator languages can be recognized with this inclusive cutpoint.

On the other hand, 2-state MCQFAs can define some binary non-regular languages with inclusive cutpoint 0, e.g. $\EQ$ \cite{BC01B,BP02}. If we allow left end-markers, then $ \EQ $ can be defined with any inclusive cutpoint. So, the phenomenon of ``discrimination'' can be quite complicated and deserves further attention. 

\bibliographystyle{plain}

\bibliography{tcs}

\newpage

\end{document}